\pdfoutput=1
\documentclass[a4paper,UKenglish,numberwithinsect]{scrartcl}
\usepackage[utf8]{inputenc}
\usepackage{amsmath,amssymb,amsthm,mathtools,xspace}

\usepackage{microtype}

\newtheorem{proposition}{Proposition}

\newtheorem{claim}{Claim}
\newtheorem{problem}{Problem}

\newtheorem{remark}{Remark}[section]
\newtheorem{definition}{Definition}[section]
\newtheorem{theorem}{Theorem}[section]
\newtheorem{lemma}{Lemma}[section]
\newtheorem{corollary}{Corollary}[section]

\newcommand{\sharpsat}{\ensuremath{\#\text{3-SAT}}}

\newcommand{\snc}{{\sf snc}\mbox{-}}
\newcommand{\sign}{{\sf sgn}}
\newcommand{\cdet}{{\sf C}\mbox{-}{\sf det}}
\newcommand{\cperm}{{\sf C}\mbox{-}{\sf perm}}
\newcommand{\near}{{\sf near}}
\newcommand{\cut}{{\sf cut}}
\newcommand{\rank}{{\sf rank}}
\newcommand{\mcoeff}{{\sf mcoeff}}
\newcommand{\sym}{{\sf Sym}}

\newcommand{\pcoeff}{{\sf pcoeff}}
\newcommand{\mc}{{\sf mc}}
\newcommand{\pc}{{\sf pc}}
\newcommand{\depth}{{\sf depth}}
\newcommand{\size}{{\sf size}}

\newcommand{\planar}{{\sf planar-}}
\newcommand{\perm}{{\sf perm}}

\newcommand{\sharpp}{\ensuremath{\#{\sf P}}\xspace}

\title{New Algorithms and Hard Instances for Non-Commutative Computation}
\author{Christian Engels\footnote{Saarland University, Department of Computer Science, Germany {\tt
            engels@cs.uni-saarland.de}} \and B.\ V.\ Raghavendra Rao\footnote{IIT Madras, Department
        of Computer Science and Engineering, Chennai, India, {\tt bvrr@cse.iitm.ac.in}}}

\begin{document}
\maketitle
\begin{abstract}
      Recent developments on the complexity of the non-com\-mu\-ta\-tive determinant and
     permanent [Chien et al.\ STOC 2011, Bl\"aser ICALP 2013, Gentry CCC 2014]  have settled the complexity of non-com\-mu\-ta\-tive determinant with respect to the structure of the underlying algebra. Continuing the research further, we look to obtain  more insights on hard instances of non-commutative permanent and determinant.

 We show that any Algebraic Branching Program (ABP) computing the Cayley permanent of a collection of disjoint directed two-cycles with distinct variables as edge labels requires exponential size. For  graphs where every connected component  contains at most six vertices, we show that evaluating the Cayley permanent over any algebra containing $2\times 2$ matrices is $\#{\sf P}$ complete.
 
Further,  we obtain efficient algorithms for computing the Cayley permanent/determinant on graphs with bounded component size, when vertices within each component are not far apart from each other in the Cayley ordering. This gives a tight upper and lower bound for size of ABPs computing the permanent of disjoint two-cycles.  Finally, we exhibit more families of non-commutative polynomial evaluation problems that are complete for $\#{\sf P}$.
 
Our results demonstrate that apart from the structure of underlying algebras, relative ordering of the variables plays a crucial role in determining the complexity of non-commutative polynomials.
 
\end{abstract}

\section{Introduction}
\subsubsection*{Background}
The study of algebraic complexity theory was initiated by Valiant in his seminal paper \cite{Val79b}
where he showed that computing the permanent of an integer matrix is $\sharpp$ complete. Since then, separating the
complexities of permanent and determinant has been the focal point of this research area which led to the
development of several interesting results and techniques. (See~\cite{Bur00,SY10} for good
surveys on these topics.)

The underlying ring plays an important role in algebraic complexity theory. While the
research  focused mainly on the  permanent vs determinant problem over fields and commutative rings there has also
been an increasing amount of interest over non-commutative
algebras. Nisan~\cite{NisanNonComm}  was the first  to consider the complexity of these two polynomials
over non-commutative algebras. He showed that any non-commutative arithmetic formula over  the
free $\mathbb{K}$ algebra   computing the permanent or determinant of an $n\times n$ matrix requires
size $2^{\Omega(n)}$ where $\mathbb{K}$ is any field.  Later on, 
this was generalized to other classes of  algebras in~\cite{CS07}. More recently, Limaye, Malod and Srinivasan~\cite{LMS15} generalized Nisan's technique to prove lower bounds against more general classes of non-commutative circuits.  Nisan's work left the problem of determining the arithmetic circuit complexity of non-commutative determinant as an open question.

 In a significant breakthrough, Arvind and Srinivasan~\cite{AS10} showed that computing the Cayley determinant is $\sharpp$ hard over certain   matrix
algebras.
 Finally this question was settled
by Bl\"aser~\cite{Bla13} who classified  such algebras. Further, Gentry~\cite{Gentry14} simplified these reductions.

\subsubsection*{Motivation}
\label{sec:motivation}
Though the studies in~\cite{AS10,Bla13} highlight the role of the underlying algebra in determining the complexity of the
non-commutative determinant they do not shed much light on the combinatorial structure of non-commutative polynomials that are
$\sharpp$ hard. One could ask: \emph{Does the hardness stem from the underlying algebra or are there inherent properties of polynomials
that make them $\sharpp$ hard in the non-commutative setting?} Our  results  in this paper indicate that   relative ordering among the variable also plays an important role in the hardness of certain non-commutative polynomials.

As a first step,  we look for   polynomials that are easier to compute  than the determinant in the commutative setting and   whose non-commutative versions are   $\sharpp$ hard. Natural candidate polynomials are the elementary
symmetric polynomials and special cases of determinant/permanent. One way to obtain special cases of
determinant/permanent would be to restrict the structure of the underlying graph.  For example, let $G$ be a directed graph consisting of $n$ cycles
$(0,1),(2,3),\dots,(2n-2,2n-1)$ of length two with self loops where each edge is labeled by a
distinct variable. The permanent of $G$, $\perm(G)$, is given by $\prod_{i=0}^{n-1}(
x_{2i,2i}x_{2i+1,2i+1} + x_{2i,2i+1}x_{2i+1,2i})$ where $x_{i,j}$ is the variable labeling of the
edge $(i,j)$. This is one of the easiest to compute but non trivial special case of permanent. 

\subsubsection{Our Results}
\label{sec:main-results}
We study the complexity of the Cayley permanent (\cperm) on special classes of graphs. We exhibit
a family of collections of disjoint two-cycles for which any algebraic branching program (ABP) computing
the $\cperm$ should have size $2^{\Omega(n)}$ (Corollary~\ref{cor:explicit-perm}).
Further, we exhibit a parameter
$\cut(G)$  (see Section~\ref{sec:lb} for the definition) for a collection $G$ of disjoint two-cycles on $n$ vertices such that any ABP computing  $\cperm(G)$ has size
$2^{\Theta(\cut(G))}$ (Theorem~\ref{thm:cut-lowerbound}).   This  makes the lower bound in Corollary~\ref{cor:explicit-perm} tight up to a constant factor in the exponent. It should be noted that our results also hold for the
case of the Cayley determinant ($\cdet$) on such graphs. We also observe that for graphs of component size greater or equal to
six the problem of evaluating $\cperm$ is $\sharpp$ complete (Theorem~\ref{thm:gentry}).
 
On the positive side, for  graphs where each strongly connected  component has at most $c$ vertices we obtain an ABP of size
$n^{O(c)}c^{\near(G)}$ computing the \cperm\  (Theorem~\ref{thm:ub-bounded}) where $\near(G)$ is a parameter (see
Definition~\ref{def:near}) depending on the labeling
of vertices on the graph.

We demonstrate a non-commutative variant of the elementary symmetric polynomial that is $\sharpp$ hard
over certain algebras (Theorem~\ref{thm:sym}). Finally, we show that computing  $\cperm$ on rank one
matrices is $\sharpp$ hard.

\subsubsection{Related Results}
The study of commutative permanent on special classes of matrices was initiated by Barvinok
\cite{Bar96} who gave a polynomial time algorithm for computing the permanent of rank one matrices over a
field. More recently, Flarup, Koiran and Lyaudet~\cite{FKL07} showed that computing the permanent of
bounded tree-width graphs can be done by polynomial size formulas. This was further extended by
Flarup and Lyaudet \cite{FL10} to other width measures on graphs. Datta et al.\ \cite{DKLM10} showed that
computing the permanent on planar graphs is as hard as the general case. 

\subsubsection{Comparison to other results}
Results reported in~\cite{AS10,Bla13,Gentry14} highlight the importance of the underlying
algebra and characterizes algebras for which $\cdet$ is $\sharpp$ hard.  In contrast, our results shed
light on the role played by the order in which vertices are labeled in a graph. For example, the commutative permanent of disjoint two-cycles has a depth three formula given by $\prod_{i=0}^{n-1}(
x_{2i,2i}x_{2i+1,2i+1} + x_{2i,2i+1}x_{2i+1,2i})$ whereas  \cperm\  on almost all orderings of vertices requires exponential size ABPs. 
 

\section{Preliminaries}
\label{sec:prelim}
For definitions of complexity classes the reader is referred to any of the standard text books on Computational Complexity Theory, e.g.,~\cite{AB09}.
Let $\mathbb{K}$ be a field and $ S=\mathbb{K}[x_{1},\dots,x_{n}]$ be the  ring of polynomials over $\mathbb{K}$ in $n$ variables. Let $R$ denote a non-commutative ring with
identity and associativity property. Unless otherwise stated,    we assume that $R$ is an algebra over $\mathbb{K}$ and contains the algebra of $n\times n$ matrices with entries from $\mathbb{K}$ as a subalgebra.  

An \emph{arithmetic circuit} is a directed acyclic graph  where  every vertex has an in-degree either zero or two. Vertices of zero in-degree  are called  {\em input} gates and are labeled by   elements in $R\cup \{x_1\ldots, x_n\}$.   Vertices of in-degree two  are called {\em internal} gates and   have their labels from $\{\times, +\}$.   An arithmetic circuit has  at least one vertex of out degree zero called an  {\em   output}  gate. We assume that an arithmetic circuit has  exactly one output gate. A  polynomial $p_g$ in $R[x_1,\ldots, x_n]$ can be associated with every gate $g$  of an arithmetic circuit defined in an inductive fashion. Input gates compute
their label. Let $g$ be an internal gate with left  child $f$ and right child $h$,  then $p_g=p_f~{\sf op}~ p_h$  where ${\sf op}$ is the label of $g$.   The polynomial computed by the circuit is the polynomial at one of the output gates and denoted by $p_C$. The size of an arithmetic circuit is the number of gates in it and is denoted by $\size(C)$. We restrict ourselves to circuits where  coefficients of the polynomials computed at every gate can be represented in at most ${\sf poly}(\size(C))$ bits. 

An {\em algebraic branching program} (ABP) is a directed acyclic graph with two special nodes $s$,
$t$ and edges labeled by variables or constants in $R$. The weight of a path is the product of the
weights of its edges. The polynomial computed by an ABP $P$ is the sum of the weights of all $s\leadsto t$ paths in  $P$, and is denoted by $p_P$.

Over a non-commutative ring,   there are many possibilities for defining the
determinant/permanent of a matrix  depending on the ordering of the variables  (see for example~\cite{AH96}). We will use the well known definitions of the
\emph{Cayley determinant} and \emph{Cayley permanent}. Let $X=(x_{i,j})_{1\le i,j\le n}$ be an $n\times n$ matrix with distinct variables $x_{i,j}$.  Then
{\small
\begin{align*}
  {\cdet}(X) = \sum_{\sigma \in S_{n}} \text{sgn}(\sigma) x_{1,\sigma(1)}\cdots x_{n,\sigma(n)}; ~~ \mbox{ and }
   \cperm(X) = \sum_{\sigma \in S_{n}} x_{1,\sigma(1)}\cdots x_{n,\sigma(n)}.
\end{align*}
}
In the above,  $S_n$ denotes the set of all permutations on $n$ symbols.
Note that $\cdet$ and  $\cperm$ can also be seen as functions taking $n\times n$ matrices with
entries from $R$ as input.  Given a weighted directed graph $G$ on $n$ vertices with weight  $x_{i,j}$ for the  edge $(i,j)\in E(G)$, the Cayley permanent of $G$ denoted by $\cperm(G)$ is the permanent of the weighted adjacency matrix of $G$. It is known that~\cite{Bur00} $\cperm(G)$ is the sum of the Cayley weights of all cycle covers of $G$.

The tensor product of two matrices $A,B\in \mathbb{K}^{n\times n}$ with entries $a_{i,j},b_{i,j}$
 is denoted by $A\otimes B$  and is given by
$$A\otimes B = \left(\begin{matrix}  a_{1,1}B& \cdots& a_{1,n}B\\
                        \vdots& \ddots& \vdots\\
                        a_{n,1}B& \cdots& a_{n,n}B
        \end{matrix}\right).
$$
  Let $P$ be an  ABP
 over disjoint sets of  variables $X\cup Y$, with $|X|=n$ and $|Y|=m$. Let $p_P(X,Y)$ be the
 polynomial computed by $P$. $P$ is said to be {\em read once certified}~\cite{MR13} in $Y$ if  there are numbers
 $0=i_0 < i_1< \dots < i_m$ where $i_m$ is at most the length of $P$  and there is a permutation $\pi\in S_m$ such that
  between layers from $i_j$ to $i_{j+1}$ no variable other than $y_{\pi(j+1)}$ from the set  $Y$  appears as a label.
  We use the following result from~\cite{MR13}. The proof given in~\cite{MR13} works only in the commutative setting, see  Appendix~\ref{app:expsum} for the non-commutative case.
 \begin{proposition}[\cite{MR13}]
\label{prop:expsum} Let $P$ be an ABP on $X\cup Y$ read-once certified in $Y$. Then the  polynomial
$
\sum_{e_1,e_2,\ldots, e_m\in\{0,1\}^{m}} p_P(X, e_1,\ldots, e_m)$
can be computed  by an ABP of size ${\sf poly}(\size(P))$.
\end{proposition}

  Let  ${\cal A}$  be  a non-de\-ter\-min\-is\-tic   $s$-space bounded algorithm  that   uses non-de\-ter\-min\-is\-tic bits in a read-once fashion and  outputs a monomial on each of the accepting paths.    We assume that a non-commutative monomial is output as a string in a write-only tape and    non-deterministic paths are represented by binary strings $e\in \{0,1\}^m$, $m\le 2^{O(s)}$.     The polynomial $p_{\cal A}$ computed by $\cal A$ is the sum of the monomial output on each of the accepting paths of ${\cal A}$, i.e., $p(x_1,\ldots, x_n) = \sum_{e} A(x_1,\ldots, x_n, e)$, where the sum  is taken over all accepting paths $e$ of ${\cal A}$, and $A(x_1,\ldots, x_n, e)$ denotes the monomial output along path  represented by $e$.    

\begin{proposition}[folklore] 
\label{lem:folklore}
Let $A(X)$ be an $s$-space bounded  non-deterministic algorithm  as above.  There is a non-commutative ABP $P$ of size $2^{O(s)}$ that computes the polynomial $p_{A}(X)$.

\end{proposition}

\section{An algorithm for Cayley Permanent}
\label{sec:algo}
In this section, we give an algorithm for $\cperm$ that is parameterized by the maximum difference between labelings of vertices in individual components.

In what follows, we identify the vertices of a graph with the set $[n]$. A directed graph $G$ on $n$
vertices is said  to have {\em  component size bounded by $c$} if every strongly connected component of $G$
contains at most $c$ vertices where $c>0$.   We assume that edges of $G$ are labeled by distinct variables.  Firstly, we define a parameter that measures the closeness of labelings in each component.
\begin{definition}
\label{def:near}
Let $G$ be  a directed graph.  The {\em nearness} parameter $\near(C)$ of a strongly connected component $C$ of $G$ is defined as 
$\near(C) = \max_{ i, j\in C} |i-j|$.
The nearness parameter of $G$ is defined as
$\near(G) = \max_{C}\near(C)$,
where the maximum is taken over the set of all strongly connected components in $G$.
\end{definition}
\begin{theorem}
\label{thm:ub-bounded}
Let $G$ be a directed graph with  component  size
bounded by $c$ and  edges labeled by distinct variables. Then there exists  an ABP of size $n^{O(c)} c^{\near(G)}$ computing  the Cayley
permanent of the adjacency matrix of $G$.
\end{theorem}
\begin{proof}
For an edge $(i,j)\in E(G)$, let $x_{i,j}$ denote the variable label on $(i,j)$. Let $A_G$ be the weighted
adjacency matrix of $G$.   Note that, the Cayley permanent of $A_G$ equals the sum of weights of
cycle covers in $G$ where the weight of a cycle cover $\gamma$  is the product of labels of edges in $\gamma$ multiplied in the Cayley order.
  
We  describe a non-deterministic small-space bounded procedure $P$ that  guesses a cycle cover $\gamma$
in $G$ and outputs the product of weights of $\gamma$  with respect to the Cayley ordering as a string of variables. Additionally,
we ensure that the algorithm $P$ uses the non-deterministic bits in a read-once fashion, and by the
closure property of ABP under read-once exponential sums (c.f.\ Proposition~\ref{prop:expsum}),  we obtain the required ABP.
Suppose $C_1, \ldots, C_r$ are the strongly connected components of $G$, sorted in the ascending order of the smallest  vertex in each component. Then any cycle cover $\gamma$ of $G$ can be decomposed into  cycle cover $\gamma_i$ of the component $C_i$. The only difficulty in computing the weight of $\gamma$ is the Cayley ordering of the variables. However, with a careful implementation, we show that this can be done in space $O(\log c\cdot  \near(G)+ \log n)$.  We represent a  cycle cover in $G$ as a permutation $\gamma$ where $\gamma(i)$ is the successor of vertex $i$ in the cycle cover represented by $\gamma$.
We begin with the description of the non-deterministic procedure $P$. Let  $T$ represent the
set of vertices $v$ in the partial  cover that is being built by the procedure where the weight of
the edge going out of $v$ is not yet output, and
$\text{pos}$ the current position going from $1$ to $n$.
\begin{enumerate}
\item Initialize $\text{pos} :=1$, $ T:=\emptyset$, $\gamma:=$ the  cycle cover of the empty graph, $f=1$.
\item  For $1\le i\le r$  repeat   steps 3 \& 4.   
\item Non-deterministically guess a cycle cover $\gamma'$ in  $C_i$, and set $\gamma=\gamma\uplus
\gamma'$, $T=T\cup V(C_i)$ where $V(C_i)$ is the set vertices in $C_i$.
\item While there is a vertex $k\in T$ with $ k=\text{pos}$ do the following:
 \begin{itemize}
  \item[ ]  Set $f=f\cdot x_{k,\gamma(k)}$;  $\text{pos}:=\text{pos}+1$; and $T:=T\setminus \{k\}$. 
 \end{itemize}
\item If $\text{pos}= n$, then  output $f$ and accept.
\end{enumerate} 
Let ${\sf Acc}(G)$ be the sum of the monomials  output by the algorithm on all accepting paths. 
\begin{claim}
\label{cla:space}
${\sf Acc}(G)= \cperm(G)$. Moreover, the algorithm $P$ uses $O(\log c\cdot\near(G)+\log n)$ space, and is read-once on the non-deterministic bits.
\end{claim}
\begin{proof}[of the Claim]
    Recall that a permutation $\gamma\in S_n$ is a cycle cover of $G$ if and only if it can be
    decomposed into vertex disjoint cycle covers $\gamma_1, \ldots, \gamma_r$ of the strongly
    connected components $C_1,\ldots, C_r$ in $G$. Thus Step~3 enumerates all possible cycle covers
    in $G$.  Also, the weights output at every accepting path are in the Cayley order.

We have $T= \{k~|~ \text{pos} <  k ~\mbox{and $k$ occurs in the components already explored} \}$.  
 Firstly, we argue that at any point in time in the algorithm, $|T|\le \near(G)+c$.  Suppose the
 algorithm  has processed components  up to $C_i$ and is yet  to process $C_{i+1}$.  Let $\mu=\max
 _{v\in T} v$. Since the components are in ascending order with respect to the smallest vertex in them, the component $C_j$ with $\mu\in C_j$ must have $\near(C_j)\ge \mu-\text{pos}$. Thus $\mu-\text{pos} \le \near(G)$. Also, just before step 3 in any iteration,  for any $v\in T$, we have $\text{pos} < v\le \mu$ and hence $|T| \le \mu-\text{pos} +c\le \near(G)+c$.
 
 Note that it is enough to store the labels of the vertices in $T$ and the choice $\gamma(v)$ made
 during the non-deterministic guess for each $v\in T$ and hence $O(|T|\log n)$
 additional bits of information needs to be stored. However, we will show that it is possible to implement
 the algorithm without explicitly remembering the vertices in $T$ and using only $O(|T|\log c)$
 additional bits in memory.  Suppose that the vertices in $T$ are  ordered as
 they
 appear in  $C_1, C_2 ,\ldots , C_r$ where vertices within a component are considered
 in the ascending order of their labels.   Let $B$ be a vector of length $\near(G)$ where each entry $B_j$ is $\log c$
 bits long which indicates the neighbour  of the $j$th vertex in $T$.  Now, we show  how to
 implement step 4 in the procedure using $B$ as a data structure for $T$.   To check if there is a
 $k\in T$ with $k=\text{pos}$, we can scan the components from $C_1, \ldots, C_i$ and check if the
 vertex assigned to $\text{pos}$ occurs in one of the components. Remember that $\gamma(k)$ is the
 successor of $k$ in the cycle cover $\gamma$. To obtain $\gamma(k)$ from $B$,
 we need to know the number $j$  of vertices $v$ that appear in components $C_1,\ldots, C_i$ such
 that $v\ge \text{pos}$ and that occur before $k$. Then $\gamma(k)=B_j$.
 Once $B_j$ is  used, we remove $B_j$ from $B$ and shift the array $B_{j+1},\ldots B_{\near(G)+c}$ by one index towards the left. 
Further, we can implement step $3$ by simply appending the information for $V(C_i)$ given by
$\gamma'$ to the right of the array $B$. We  require at most $O(c\log n)$ bits of space  guessing a cycle cover $\gamma_i$ for
component $C_i$  which can  be re-used after the non-deterministic guessing of $\gamma_i$ is complete.
Thus the overall space requirement of the algorithm is
bounded by $O(\log c\cdot (\near(G)+c) + c\log n)$.
 \end{proof}
By Proposition~\ref{lem:folklore}, we get an ABP $P$ computing a polynomial $p_G(X,Y)$  such that 
$\cperm(G)= \sum_{e_1,\ldots, e_m\in\{0,1\}} p_G(X,e)$, $m = O(c\log n)$.
Combining the above algorithm with the  closure property of algebraic branching programs over
read-once variables given by Proposition~\ref{prop:expsum},  we get a non-commutative arithmetic branching program computing $\cperm(G)$.  It can be seen that size of the resulting branching program  
is at most $ m 2^{O((c\log c+\log c\cdot\near(G))+c\log n)}=  n^{O(c)}\cdot c^{\near (G)}$ for large enough $n$.
\end{proof}
%
\begin{corollary}
\label{cor:cdet}
Let $G$ be as in Theorem~\ref{thm:ub-bounded}. There is an ABP of size $n^{O(c)}c^{\near(G)}$ computing the Cayley determinant of $G$.
\end{corollary}
\begin{proof}
    The argument is the same as in Theorem~\ref{thm:ub-bounded} except that now the non-deterministic
algorithm given in the proof of Theorem~\ref{thm:ub-bounded} also needs to compute the sign of the
monomial being output. Let $C_1,\ldots, C_r$ be the strongly connected components of $G$. Then the sign of the permutation corresponding to a cycle cover $\tau$ of $G$ is the product of signs of the corresponding cycle covers of $C_i$.   Thus it is enough to modify the algorithm given in the proof of Theorem~\ref{thm:ub-bounded}  to output the sign of the cycle cover chosen for $C_i$, the remaining arguments are  the same.
\end{proof}

\section{Unconditional Lower Bound}
\label{sec:lb}
We now show that any branching program computing the non-commutative permanent of directed graphs
with component size $2$ must be of exponential size. This shows that the upper bound in
Theorem~\ref{thm:ub-bounded} is tight up to a constant factor in the exponent, however, with a
different but related parameter. All our lower bound results hold for  free algebras over any field $\mathbb{K}$.

 Our proof crucially depends on Nisan's (\cite{NisanNonComm}) partial derivative technique. We begin
 with some notations following his proof.  Let $f$ be a non-commutative degree $d$ polynomial in $n$
 variables.  Let $B(f)$ denote the smallest size of a non-commutative ABP computing $f$.  For
 $k\in \{0,\ldots, d\}$ let $M_k(f)$ be the matrix with rows indexed by all possible sequences
 containing $k$ variables and columns indexed by all possible sequences containing $d-k$ variables
 (repetitions allowed).  The entry of $M_{k}(f)$ at $(x_{i_1}\ldots x_{i_k}, x_{j_1}\ldots x_{j_{d-k}})$ is
 the coefficient of the monomial $x_{i_1}\cdots x_{i_k}\cdot x_{j_1}\cdots x_{j_{d-k}}$ in
 $f$. Nisan established the following result:
\begin{theorem}\cite{NisanNonComm}
\label{thm:nisan}
For any homogeneous polynomial $f$ of degree $d$,
$$B(f)= \sum_{k=0}^d {\sf rank}(M_{k}(f)).$$
\end{theorem}  

We prove lower bounds for the Cayley permanent of graphs with every strongly connected component of
size exactly $2$, i.e., each strongly connected component being a two-cycle with self loops on the vertices. Note that any
collection of $n/2$ vertex disjoint two-cycles can be viewed as a permutation $\pi\in S_n$
consisting of disjoint transpositions and that $\pi$ is an involution.  Conversely, any  involution
$\pi$ on $n$ elements represents a  graph $G_\pi$ with connected component size $2$.

For a permutation $\pi\in S_n$ let the {\em cut at $i$} denoted by $C_i(\pi)$ be the set of pairs $(j,\pi(j))$  that cross $i$, i.e.,
$C_{i}(\pi)= \{(j,\pi(j))~|~  i\in[j,\pi(j)]\cup [\pi(j), j]\,\}$.
The {\em cut} parameter $\cut(\pi)$ of $\pi$ is defined as 
$\cut(\pi) = \max_{1\le k\le n} |C_k(\pi)|$.
Let $G$ be a collection of vertex disjoint 2-cycles denoted by $(a_1,b_1), \ldots, (a_{n/2},b_{n/2})$ where $n$ is even.
The corresponding involution is $\pi_G= (a_1,b_1)\cdots (a_{n/2},b_{n/2})$. By abusing the notation a bit, we let $\cut(G)=\cut(\pi_G)$.   Without loss of generality, assume that $a_i<b_i$, and $a_1<a_2<\dots <a_{n/2}$.  
Firstly, we note that $\cut(\pi)$ is bounded by $\near(G)$. 
\begin{lemma}
\label{lem:cut-near}
For any collection of disjoint 2-cycles $G$ on $n$ vertices,  $\cut(\pi)\le \near(G)$ where $\pi$ is the involution represented by $G$.
\end{lemma}
\begin{proof}
Suppose $\cut(\pi)=r$, and $1\le i\le n$ be such that $|C_i(\pi)|=r$. Let $(a,b) \in G$  where  $a$
 is the least value with $a<i $ and $b>i$ be the maximum such value. Then
$b-a \ge  2 |C_i(\pi)|/2 =r$. This concludes the proof.
\end{proof}

 Further, we note  that the upper bound given in Theorem~\ref{thm:ub-bounded} holds true even if we
 consider $\cut(G)$ instead of  $\near(G)$. 
\begin{lemma}\label{lem:up-involution}
    Let $G$ be a collection of disjoint 2-cycles and self loops where every edge is labeled by a
    distinct variable or a constant from $R$. Then there is an ABP of size $2^{{O}(\cut(G))} n^2$
    computing the Cayley permanent of $G$.
\end{lemma}
\begin{proof}
    The algorithm is the same as in Theorem~\ref{thm:ub-bounded}. We only need to argue the space bound
as in the claim in Theorem~\ref{thm:ub-bounded}.   First note that either $a_i=i$, or $i$ has already occurred in one of the  involutions $(a_a,b_1),\ldots, (a_{i-1}, b_{i-1})$. When the algorithm processes the component corresponding to the involution $(a_i, b_i)$, it needs to remember the outgoing edge chosen for $b_i$ (either the self loop or the edge $b_i\to a_i$).  Thus at any stage, the number of edges that needs to be stored is bounded by $t=\max_k~|C_k(\sigma)|$. The rest of the arguments are exactly the same as in Theorem~\ref{thm:ub-bounded}. 
\end{proof}

\begin{lemma}\label{lem:tensor-lowerbound}
    Let $G$ be a collection of $\ell$ disjoint two-cycles described by the involution $\pi$ and self loops
    at every vertex with edge labeled by distinct variables.  Then  $M_\ell(\cperm(G))$ contains $I_{2}^{\otimes t}$ as a sub-matrix
    where $t=\max_{k} |C_{k}(\pi)|$, $A^{\otimes t}$ is the tensor product of $A$ with itself $t$ times and $I_2$ is the $2\times 2$ identity matrix.
\end{lemma}
\begin{proof}
    Let $k\in [\ell]$, and $m=|C_k(\pi)|\le \ell$. Let
    $C_k(\pi) = \{(a_{i_1}, b_{i_1}), \ldots, (a_{i_m}, b_{i_m})\}$ be such that
    $a_{i_j}\le k \le b_{i_j}$ for all $j$.  Let $G_k$ be the graph restricted to involutions in
    $C_k(\pi)$. By induction on $m$, we argue that $M_m(\cperm(G_k))$ contains
    $I_{2}^{\otimes m}$ as a sub-matrix. The lemma
    would then follow since $M_m(\cperm(G_k))$ is itself a sub-matrix of
    $M_\ell(\cperm(G))$.

    We begin with $m=1$ as the base case.  Consider the transposition $(a_{i_j}, b_{i_j})$, with
    $a_{i_j}\le k\le b_{i_j}$. The corresponding two cycle has four edges. Let $f_{i_j}$ be the
    Cayley permanent of this graph then $M_1(f_{i_j})$ has the $2\times 2$ identity matrix as a
    sub-matrix.  Let us dwell on this simple part. For ease of notation let the variables
    corresponding to the self loops be given by $x_{a},x_{b}$ for $(a_{i_{j}},a_{i_{j}})$ and
    $(b_{i_{j}},b_{i_{j}})$ respectively and the edge $(a_{i_{j}},b_{i_{j}})$ by $x_{(a,b)}$ and the
    edge $(b_{i_{j}},a_{i_{j}})$ by $x_{b,a}$. Now our matrix has monomials $x_{a},x_{a,b}$ as rows
    and $x_{b},x_{b,a}$ as columns. We can ignore the other orderings as these will always be
    zero. As the valid cycle covers are given by $x_{a}x_{b}$ and $x_{a,b}x_{b,a}$ the proof is
    clear.
 
    For the induction step, suppose $m > 1$.  Suppose $a_1 < a_2 < \dots < a_m$. Let $G_k'$ be the
    graph induced by $C_k(\pi)\setminus (a_1,b_1)$. Let
    $M'=M_{m-1}(\cperm(G_k'))$. The rows of $M'$ are labeled by monomials
    consisting of variables with first index $\le k$ and the columns of $M'$ are labeled by
    monomials consisting only of variables with first index $>k$.  Let
    $M=M_{m}(\cperm(G_k))$. $M$ can be obtained from $M'$ as follows: Make two
    copies of the row labels of $M'$, the first one with monomials pre-multiplied by $x_{a_1,a_1}$,
    and the second pre-multiplied by $x_{a_1,b_1}$. Similarly, make two copies of the columns of
    $M'$, the first by inserting $x_{b_1,b_1}$ to the column labels of $M'$ at appropriate position,
    and then inserting $x_{b_1,a_1}$ similarly. Now, the matrix $M$ can be viewed as
    two copies of $M'$ that are placed along the diagonal. 
     Thus  
    $M= M'\otimes I_2$, combining this with  Induction Hypothesis completes the proof.   
\end{proof}
\begin{remark}
It should be noted that the ordering of the variables is crucial in the above argument. If $a_{1},b_{1}<k$ in the   above,   then ${\sf rank}(M)={\sf rank}(M')$.   
\end{remark}

\begin{theorem}\label{thm:cut-lowerbound}
   Let $G$  be a  collection of disjoint two cycles described by the involution $\pi$ and self loops at every vertex, with edges labeled by distinct variables. 
   Then any  non-commutative ABP computing the Cayley permanent on  $G$ has size at least $2^{\Omega(\cut(G))}$. 
\end{theorem}
\begin{proof}
   It is enough to argue that for every $k$, there is an $\ell$ with
    ${\sf rank}(M_\ell(f))\ge 2^{\Omega(|C_k(\pi)|)}$, then the claim follows from
    Theorem~\ref{thm:nisan} as the lower bound is given by the sum over all $\ell$.  Let
    $\ell = |C_{k}(\pi)|$, and suppose $(a_{i_1}, b_{i_1}),\ldots,(a_{i_\ell}, b_{i_\ell})$ are the
    transpositions crossing $k$. Let $G'$ be the sub-graph of $G$ induced by the vertices corresponding
    to the transpositions above.  Let $f= \cperm(G') $.  Applying Lemma~\ref{lem:tensor-lowerbound} on
    $G'$ we conclude that $M_{\ell}(f)$ has $I_{2}^{\otimes |C_{k}(\pi)|}$ as a sub-matrix, i.e., the
    identity matrix of dimension $2^{|C_{k}(\pi)|}\times 2^{|C_{k}(\pi)|}$.  Note that $ f$ can be
    obtained by setting weights of the self loops of vertices not in $G'$ to zero, and setting the
    remaining variables to $1$.  Moreover, the matrix $M_{\ell}(f)$ is a sub matrix of
    $M_{\ell}(\cperm(G))$ obtained by relabeling the rows and columns as per the substitution mentioned
    above, and removing rows and columns that are zero. We conclude
    $\rank(M_{\ell}(\cperm(G))) \ge\rank(M_{\ell}(f))\ge 2^{|C_k(\pi)|}$.
\end{proof}

Let $\pi = (a_1,b_1) \cdots (a_{n/2},b_{n/2})$, $a_1<a_2< \dots<a_{n/2}$
be an involution.  Then  $G_\pi$   is the set of $2$-cycles
$(a_1,b_1), \dots, (a_{n/2},b_{n/2})$ and self loops at every vertex.
\begin{corollary}\label{cor:explicit-perm}
    Let $G$ be a collection of disjoint two cycles described by the involution $\pi$ and self loops
    at every vertex, with edges labeled by distinct variables.  Then
    $B(\cperm(G))\in 2^{\Theta(\cut(G))}$.  Further, there exists a graph $G$ with
    $\cut(G)=\Theta(n)$.
\end{corollary}
\begin{proof}
    The first part follows immediately from Lemma~\ref{lem:up-involution} anf
    Theorem~\ref{thm:cut-lowerbound}. For the second statement, consider the  involution $\pi$ with 
    $\pi(i) = n/2+i$. It can be seen that $\max_k |C_k(\pi)| = n/2$. Thus by
    Theorem~\ref{thm:cut-lowerbound} we have any ABP computing the Cayley permanent of $G$ is of
    size $2^{\Omega(n)}$. Since $\cut(\pi)\le n$ for any graph, by Lemma~\ref{lem:up-involution} the
    result follows.
\end{proof}


Finally, we have, 

\begin{theorem}
\label{thm:cut-average}
For all   but a $1/\sqrt{n}$ fraction of  graphs $G$ with connected component size 2, any ABP computing  the \cperm\ on $G$ requires size $2^{\Omega(n)}$.

\end{theorem}

As before, let $n=2m$ be even. Then an involution $\pi$ on $\{1,\ldots, n\}$ with $\pi(i)\neq i$
represents a collection of $m$  intervals
\[I_\pi=\{[i,\pi(i)]~|~ 1\le i\le n,~ i<\pi(i)\}.
\]
Let $H_\pi$ be the interval graph formed by the intervals in $I_\pi$.   
\begin{lemma} 
\label{lem:involution}
Let $\pi$ be an involution and $H_\pi$ be the interval graph as defined above. Then $\cut(\pi)\ge l/n$ where $l$ is the number of edges in $H_\pi$.
\end{lemma}
\begin{proof}
For every edge $(a,b)$ in $H_\pi$, the corresponding intervals $I_a=[i,\pi(i)]$ and $I_b=[j,\pi(j)]$
have non empty intersection. Suppose $i<j$, then $j\in [i,\pi(i)]$. (In the case when $j>\pi(j)$, we
have $\pi(j)\in [i,\pi(i)]$. Other cases can be handled analogously.) Thus every edge in $H_\pi$
contributes at least one distinct interval $[i,\pi(i)]$ with $i\le k\le \pi(i)$, i.e., it
contributes a value to $C_k(\pi)$. Then $l\le \sum_{k}|C_k(\pi)|$. This concludes the
proof.
\end{proof}

Scheinerman~\cite{Sch88} showed that, random interval graphs have $\Omega(n^2)$ edges with high probability, i.e.,
\begin{theorem}\cite{Sch88}
Let $H_\pi$ be an interval graph where $\pi$ is an involution on $[n]$ chosen uniformly at random. Then $H_\pi$ has at least $n^2/3-n^{7/4}$ edges  with probability at least $1-1/\sqrt{n}$.    
\end{theorem}
\begin{corollary}
\label{cor:random-involution}
For an involution $\pi$ on $[n]$ chosen uniformly at random, we have $\cut(\pi)=\Omega(n)$ with probability $1-1/\sqrt{n}$.
\end{corollary}
Theorem~\ref{thm:cut-average} now follows.

\section{\sharpp completeness}
\label{sec:complete}
In this section, we show multiple  hardness results for simple polynomials over certain classes of non-commutative algebras.
We give a $\sharpp$ completeness result for specific graphs of component size at most six.  The
completeness result is obtained by a careful analysis of the parameters in the reduction from
$\#SAT$ to non-commutative determinant given recently by Gentry~\cite{Gentry14} and the small
modification we will do to make this proof work for the Cayley Permanent. 

\begin{theorem}
\label{thm:gentry}
Let $R$ be a division algebra over a field $\mathbb{K}$ of characteristic zero containing the algebra of
$2\times 2$ matrices over $\mathbb{K}$.
Computing the Cayley Permanent on graphs with component
size 6 with edges labeled from $R$ is $\sharpp$ complete.
\end{theorem}
\begin{proof}
    It is known that counting the number of satisfying assignments in a $2$-CNF formula where every
    variable occurs at most three times is already $\sharpp$ complete (\cite{Roth96}).  Let $\phi$
    be a 2-CNF where every variable occurs at most three times with $k$ clauses. We complete the
    proof by a careful analysis of the reduction given in Theorem~6 of \cite{Gentry14} applied to
    $\phi$.

    Lemma~\ref{lem:noncom:gentrylem2} gives a product program of length $2^{2}+2^{2-1}-2=4$ for
    computing a disjunction of two literals.  In fact the program for $x_{1}\vee x_2$ is given by
    $(1,(s,I_{2})),(2,(r,I_{2})),(1,(s,I_{2})),(2, (r^{-1},I_{2}))$ where $I_{2}$ is the $2\times 2$
    identity matrix in $R$.

    Let $t$ be a $2\times 2$ matrix as in Theorem~\ref{thm:noncom:gentrythm6}, namely $t$ has one in the
    upper left corner and zeros elsewhere.  Suppose $P_{c}$ is the product program as given above
    for the clause indexed by $c$ for $1\le c\le k$. Then product program for $\phi$ is given by
    $$
    P=\left(\prod_{1\le c\le k}t\cdot P_{c}\right) t.
    $$
    
    This immediately shows if every variable occurs at most three times in $\phi$, the product
    program above reads a bit of the input at most 6 times.  Let $\mathcal{I}_{\ell}$ as before
    where $1\le \ell\le 4k$. We have $\mathcal{I}_\ell\le 6$ for all $1\le \ell\le 4k$ by the above
    argument. Let $\mathcal{I}_{\ell}$ have the elements
    $i_{\ell,1},\dots,i_{\ell,|\mathcal{I}_{\ell}|}$.  Let $\pi_{0}$ be the identity permutation and
    $\pi_{1}(i_{\ell,\kappa})=i_{\ell,\kappa +1\mod |\mathcal{I}_{\ell}|}$.  Define the following
    permuted ``block barber pole'' (\cite{Gentry14}) matrix.
    \[
        M[i,j] = \begin{cases}
            a_{i,b}         &\text{ if $j=\pi_{b}(i)$}\\
            0              &\text{ otherwise.}
        \end{cases}
    \]
    While Gentry has in the cell $a_{i_{l,1},i_{l,2}}$ a different factor,
    $(-1)^{|\mathcal{I}_{\ell}|-1}$ to be precise, this factor is only canceling the sign of the
    determinant as can be seen in the proof of Theorem~\ref{thm:noncom:gentrythm5}.  As the rows and
    columns for $\mathcal{I}_{i},\mathcal{I}_{j}$ for $i\neq j$ are disjunct this matrix corresponds
    to cycles of length $|\mathcal{I}_{\ell}|\le 6$. This concludes the proof.
\end{proof}
It is known that   computing the commutative permanent of the weighted adjacency matrix of a planar graph is as hard as the general case~\cite{DKLM10}.  We observe that the reduction in~\cite{DKLM10} extends to the   non-commutative case.  
\begin{theorem}
\label{thm:planar}
 ${\cperm}\le_{m}^{p} {\planar\cperm}$; and
 ${\cdet}\le_{m}^{p} {\planar\cdet}$.
 Moreover, the above reductions work over any non-commutative algebra.
\end{theorem}
\begin{proof}
    The proof is essentially the same as in~\cite{DKLM10}. We give a brief sketch here for the sake of completeness. Let $G$ be a weighted digraph. Consider an arbitrary embedding $\mathcal{E}$ of $G$. Obtain a new graph  by changing the graph as follows:
    \begin{itemize}
        \item For each pair of edges $(u,v)$ and $(u',v')$
        that cross each other in the embedding $\mathcal{E}$, do the following:
        \item introduce two new vertices $a$ and $b$; and  
        \item  new edges $\{(a,b), (b,a), (u', a), (a,v), (u,b), (b,v')\}$
        replacing $(u,v)$ and $(u',v')$.
    \end{itemize}
    Note that any of the iterations above do not introduce any new crossings, and hence the process terminates after at most $O(n^2)$ many steps where $n$ is the number of vertices in $G$.
    Weight of $(u, v)$ is given to $(v,a)$ and $(u', v')$ is given to $(v',b)$.   The remaining edges get the weight $1$.
    By the construction, we can conclude that 
    $\cperm(G)= \cperm(G')$ and $\cdet(G)=\cdet(G')$. 
\end{proof}

We demonstrate some more families of  polynomials  whose commutative variants are easy but certain non-commutative variants are as hard as the permanent polynomial. 
We begin with a non-commutative variant of the elementary symmetric polynomial.
The elementary symmetric polynomial of degree $d$, $\sym_{n,d}$ is
given by
 $\sym_{n,d}(x_1,\ldots, x_n)= \sum_{\substack{S\subseteq [n],~|S|=d}}\prod_{i\in S}x_i.$
 There are several non-commutative variants of the above polynomial. The first one is  analogous to the Cayley permanent, i.e.,
${\sf Cayley-}\sym_{n,d}= \sum_{S=\{i_1 < i_2< \cdots< i_d \}} \prod_{j=1}^d x_{i_j}.$
 It is not hard to see that the above mentioned non-commutative version of ${\sf Cayley-}\sym_{n,d}$
 can be computed by depth 3 non-commutative circuits for every value of $d\in [n]$. However, the above definition is not satisfactory, since it is not invariant under permutation of variables, which is the inherent property of elementary symmetric polynomials. We define a variant of non-commutative elementary symmetric polynomial which is invariant under the permutation of variables.
 $${\sf nc-}{\sym_{n,d}}(x_1,\ldots,x_n)\stackrel{\triangle}{=} \sum_{\{i_1,\ldots, i_d\}\subseteq[n]} \sum_{\sigma\in S_d} \prod_{j=1}^d x_{i_{\sigma(j)}}.$$
We show that with coefficients from the algebra of $n\times n$ matrices allowed, ${\sf nc-}{\sym_{n,d}}$  cannot be computed by
polynomial size circuits unless ${\sf VP}={\sf VNP}$. We need the following definition introduced in~\cite{AJS09,AS10}.
\begin{definition}
    The \emph{Hadamard product} between two polynomials $f=\sum_{m}\alpha_{m}m$ and
    $g=\sum_{m}\beta_{m}m$, written as $f\odot g$, is defined as
    $f\odot g = \sum_{m}\alpha_{m} \beta_{m}m$.
\end{definition}

\begin{theorem}
\label{thm:sym}
Over any $\mathbb{K}$ algebra $R$ containing the $n\times n$ matrices as a sub-algebra,
${\sf nc-}{\sym_{n,n}}$ does not have polynomial size arithmetic circuits unless $\perm_n \in {\sf VP}$.
\end{theorem}
\begin{proof}
Suppose that ${\sf nc-}{\sym_{n,n}}$ has a circuit $C$ of size polynomial in $n$. We need to show that
$\perm \in {\sf VP}$.
Let $X=(x_{i,j})_{1\le i,j\le n}$
be matrix of variables, and $y_1,\ldots, y_n$ be distinct variables different from $x_{i,j}$. 
In the commutative setting, it was observed in~\cite{Gat87} that $\perm(X)$ equals the coefficient of $y_1\cdots y_n$ in the polynomial
\begin{equation}
P(X,Y)\stackrel{\triangle}{=}\prod_{i=1}^n\left( \sum_{j=1}^n x_{i,j}y_j\right)
\label{eq:hammon}
\end{equation}
over the polynomial ring $\mathbb{K}[x_{1,1},\ldots, x_{n,n}]$.  However, the same cannot be said in the case of
non-commuting variables. 
If $x_{i,j}y_k =y_k x_{i,j}$ for $i,j,k\in [n]$, then  in the non-commutative development of
(\ref{eq:hammon}), the sum of coefficients of \emph{all} permutations of the monomial  $y_1\cdots
y_n$  equals $\perm(X)$ i.e., the commutative permanent. Hence the value $\perm(X)$ can be extracted using a Hadamard product with ${\sf nc\mbox{-}}\sym_{n,n}(y_1,\ldots, y_n)$ and then substituting $y_1=1,\ldots, y_n=1$. 
However, we cannot assume $x_{i,j}y_k =y_k x_{i,j}$, since the Hadamard product may not be computable under this assumption.
Let $\ell =\sum_{i,j}x_{i,j}$.  Now we argue that
$ \perm(X)= ({\sf nc}\mbox{-}\sym_{n,n}(\ell y_1,\ldots, \ell y_n) \odot P )(y_1=1,\ldots, y_n=1)$. 
     Given a permutation $\sigma\in S_n$, there is a unique monomial
     $m_{\sigma}=x_{1,\sigma(1)}y_{\sigma(1)}\cdots x_{n,\sigma(n)}y_{\sigma(n)}$ in $P$ containing
     the variables $y_{\sigma(1)},\ldots, y_{\sigma(n)}$ in that order.  Thus taking Hadamard
     product with $P$ filters out all monomials but $m_\sigma$ from the term
     $\prod_{i=1}^n\ell y_{\sigma(i)}$.  The monomials where a $y_j$ occurs more than once are
     eliminated by ${\sf nc}\mbox{-}\sym_{n,n}(\ell y_1,\ldots, \ell y_n)$. Thus the only monomials
     that survive in the Hadamard product are of the form $m_\sigma$, $\sigma\in S_n$. Now
     substituting $y_i=1$ for $i\in [n]$  we get $ \perm(X)= ({\sf nc}\mbox{-}\sym_{n,n}(\ell y_1,\ldots, \ell y_n) \odot P )(y_1=1,\ldots, y_n=1)$.

  Note that  the polynomial
$P(X,Y)$ can be computed by an ABP of size $O(n^2)$. Then, by~\cite{AJS09,AS10}, we obtain an arithmetic circuit $D$ of size $O(n^2\size(C))$ that computes the polynomial ${\sf nc-}{\sym_{n,n}}\odot P$. Substituting $y_1=1,\ldots, y_n=1$ in $D$ gives the required arithmetic circuit for $\perm(X)$.  
\end{proof}
 Further,  let  $\snc\sym_{n,n} (x_1,\ldots, x_n)\stackrel{\triangle}{=} \sum_{\sigma\in S_n}\sign(\sigma)\prod_{i=1}^n x_{i,\sigma(i)}.$ We have
\begin{corollary}\label{cor:sign-sym}
    Over a $\mathbb{K}$ algebra containing the  algebra of $n\times n$ matrices, $\snc\sym_{n,n}$ is not in ${\sf VP}$ unless $\cdet$  has polynomial size arithmetic circuits.
\end{corollary}
While similar, this result is unrelated with the result of \cite{HWY210} as their ordered polynomials
is more related to ${\sf Cayley-}\sym_{n,d}$ than ${\sf nc-}{\sym_{n,n}}$.

Barvinok~\cite{Bar96} showed that  computing the permanent of an integer matrix of constant rank can be done in strong polynomial time. In a similar spirit, we explore the complexity of computing the Cayley permanent of bounded rank matrices with entries from $\mathbb{K}\cup\{x_1,\ldots, x_n\}$.  We consider the following notion of rank for matrices with variable entries. 
Let $A\in (\mathbb{K} \cup \{x_1,\ldots, x_n\})^{n \times n}$.   Then $
{\sf row}\mbox{-}{\sf rank}(A)= \max_{a_1,\ldots, a_n\in \mathbb{K}} {\sf rank}(A|_{x_1=a_1,\ldots, x_n=a_n})$. The column rank of $A$ is defined  analogously. As opposed to the case of the commutative permanent,
for any  algebra $R$  containing the algebra of $n\times n$ matrices over $\mathbb{K}$, we have:
\begin{corollary}\label{cor:boundedrank}
    $\cperm$ and $\cdet$ of rank one matrices with entries from $\mathbb{K}\cup \{x_{1},\dots,x_{n}\}$
    over any $\mathbb{K}$ algebra does not have polynomial size arithmetic circuits unless $\perm \in
    {\sf VP}$.
\end{corollary}
\begin{proof}
We will argue the case of $\cperm$.
Let $x_1,\ldots, x_n$ be non-commuting variables. Consider the matrix $A$ with $A[i,j]= x_j$, $1\le
i,j\le n$. $A$  has rank one over $\mathbb{K}$. We then have
${\sf nc-}{\sym_{n,n}}(x_1,\ldots ,x_n)= {\cperm}(A)$. 
The result now follows by applying Theorem~\ref{thm:sym}.
For $\cdet$, we can use Corollary~\ref{cor:sign-sym} in place of Theorem~\ref{thm:sym} in the argument above.
\end{proof}

\section{Computational problems on non-commutative circuits}

\subsection*{Computing Coefficients}
In this section we consider  various computational problems on arithmetic circuits, restricted to the non-commutative setting. We start with the problem of computing  the coefficient of a given monomial in the polynomial computed by an arithmetic circuit.  
In the commutative setting, the problem lies in the second level of the counting hierarchy~\cite{FMM12} and is known to be hard for $\#P$~\cite{Mal07}. It was first seen 
in~\cite{AMS08} that $\mcoeff$ is easy to compute in the non-commutative case.
We  provide a different proof of the fact as it is useful in the arguments used later in this section.

\begin{problem}[Monomial Coefficient(\mcoeff)]
{\em Input}: A non-commutative arithmetic circuit $C$, a non-commutative monomial $m$ of degree $d$.\\
{\em Output}: The coefficient of monomial $m$ in the polynomial computed by $C$.
\end{problem}
\begin{theorem}
\label{thm:mcoeff}\cite{AMS08}
$\mcoeff$ is in {\sf P}.
\end{theorem}
\begin{proof}
Suppose that the monomial $m =x_{j_1}\cdots x_{j_d}$  and is given as  an ordered listing of variables.  Let $f$ be a non-commutative polynomial. Then we have the following recursive formulation for the coefficient function 
${\mc}:\mathbb{K}\{x_1,\ldots, x_n\}\times {\cal M} \to \mathbb{K}$, where ${\cal M}$ is the set of all non-commutative monomials  in variables $\{x_1\ldots, x_n\}$.
\begin{eqnarray}
   \mc(f, m) = \begin{cases} \alpha &\text{if $f=\alpha y_j$ and $m=y_j$},\\
                             0 &\text{if $f=\alpha y_j$ and $m=y_i, i\neq j$ },\\
                             \mc(g, m)+ \mc(h) &\text{if $f= g+h$},\\
                             \sum_{\ell=1}^{d+1} \mc(g, m_\ell)\times \mc (h, m_{\ell}') &\text{if $f= g\times h$}.
               \end{cases}
   \label{eq:mc}
\end{eqnarray}
where $m_\ell =x_{i_1}\cdots x_{i_{\ell-1}}$ and $m_{\ell}' = x_{i_\ell}\cdots x_{i_d}$.
However, if we apply the above recursive definition on the circuit $C$ in a straightforward fashion,  the time required to compute $\mc(f, m)$ will be $d^{O(\depth(C))}$, since $\depth(C)$ could be as big as $\size(C)$, the running time would be exponential.  However, we can have a more careful implementation of the above formulation by allowing a little more space.

 For $\ell <k \in [1,d]$, let $m_{\ell,k}= x_{i_{\ell+1}}\cdots x_{i_{k}}$, and $M = \{m_{\ell,k}~|~  0\le \ell \le d-1, 0\le k\le d \}$. 
Consider a gate $v$ in the circuit $C$. Note that  in the process of computing  $\mc(f, m)$, we require only the values  from the set $ M_v=\{\mc(p_v, m')~|~ m'\in M\}$, where $p_v$ is the polynomial computed at $v$.  Thus it is enough to compute and maintain the values $\mc(p_v,m'), m'\in M$ in a bottom up fashion.   For the base case,
compute the values for  polynomials computed at  a  leaf gate $v$  as follows, let $m'\in M$ and $\alpha \in R$
\begin{eqnarray*}
   \mc(p_v, m') = \begin{cases} \alpha~\text{ if $p_v=\alpha\in R$, and $m'=\emptyset$},\\
                               \alpha~\text{ if $p_v=\alpha x_j, \alpha\in R$, and $m'=x_j$},\\
                               0~\text{ otherwise.}   
                   \end{cases}
\end{eqnarray*}    
For other nodes, we can apply the recursive formula given in (\ref{eq:mc}).
If $p_v$ = $p_{v_1}+ p_{v_2}$, then the value $\mc(p_v, m')$ can be computed using (\ref{eq:mc}) as the values 
$ \mc(p_{v_1},m')$ and $\mc(p_{v_2},m')$ are available by induction.
If $p_{v}= p_{v_1}\times p_{v_2}$, then  by induction, the values $\mc(p_{v_i}, m'' )$  are available for prefix and suffix of the monomial $m'$, as every such monomial occurs as $m_{i,j\in M}$ for some $i<j$. Now, $\mc(p_v, m')$ can be computed by (\ref{eq:mc}). 
For the space bound, the  algorithm uses $O(d^2)$ registers   for each gate in $C$ and hence the overall space  used is $O(d^2\size(C))$  many registers.  For a given monomial $m'\in M$, at most $d$ arithmetic operations are required in the worst case. Thus, the number of arithmetic operations is  bounded by $O(d^3\size(C))$. 
\end{proof}
\subsection*{Coefficient function as a polynomial}
In the commutative setting, the coefficient function of a given polynomial can be represented as a
polynomial~\cite{Mal07}. Thus it is desirable to study the arithmetic circuit complexity of
coefficient functions. However, over non-commutative rings, we need a carefully chosen
representation of monomials to obtain an arithmetic circuit that computes the coefficient function for a given polynomial with small circuits.
In the  proof of Theorem~\ref{thm:mcoeff},  we have  used an ordered listing of variables as a representation of the monomial $m$. Here we use a vector representation for non-commutative monomials of a given degree $d$.
Let $Y= \{y_{1,1}, \ldots, y_{1,n}, y_{2,1}, \ldots y_{d,n}\}$ be a set of $nd$ distinct variables, and let $\tilde{Y}_i = (y_{i,1},\ldots, y_{i
,n})$.  The vector of variables $\tilde{Y}_\ell$ can be seen as representing the characteristic vector of $x_{i_\ell}$, i.e., $y_{i, i_{\ell}} 
=1$, and $y_{i,j}=0, \forall j\neq i_{\ell}$. In essence, $y_{i,j}$ stands for the variable $x_j$ at the $i$-th position in the monomial.
Let $f(x_1,\ldots, x_n)$ be a polynomial of degree $d$, then we can define the coefficient polynomial $\pc_{f}(Y)$ as
$$
   \pc_f(Y)= \sum_{D=1}^d\ \sum_{(i_1,\ldots, i_D)\in [n]^D}\ \prod_{\ell=1}^D\left[\mc(f, x_{i_1}\cdots x_{i_D}) y_{\ell, i_{\ell}}
   \prod_{j\neq k} (1-y_{\ell, j} y_{\ell,k}) \right]. 
$$

\begin{theorem}
\label{thm:pc}
For any non-commutative polynomial  $f$ that can be computed by a polynomial size arithmetic circuit, $\pc_f(Y)$ has a polynomial size arithmetic circuit.
\end{theorem}
\begin{proof}[Proof of Theorem~\ref{thm:pc}]
We will apply (\ref{eq:mc}) to obtain an arithmetic circuit computing the polynomial $\pc_f(Y)$.
Let $C$ be an arithmetic  circuit of size $s$, computing $f$. By induction on the structure of $C$, we construct a circuit $C'$ for $\pc_f(Y)$. Note that, it is enough to compute homogeneous degree $D$ components $[\pc_f(Y)]_D$ of $\pc_f(Y)$, where
$$
   [pc_f(Y)]_D = \sum_{(i_1,\ldots, i_D)\in [n]^D}\ \prod_{\ell=1}^D\left[\mc(f, x_{i_1}\cdots x_{i_D}) y_{\ell, i_{\ell}}
   \prod_{j\neq k} (1-y_{\ell, j} y_{\ell,k}) \right].
$$
Let $Y^{i,j}$ denote the set of variables in the vectors $\tilde{Y}_{i+1}, \ldots, \tilde{Y}_j$. 
 In the base case, we have  $C=\gamma\in \{x_1,\ldots,x_n\}\cup R$. Then the  all of the homogeneous components of 
$\pc_f(Y)$ can be described as follows.
\begin{align*}
   [\pc_f(Y)]_0 \, =&      \begin{cases}
                           \gamma &\text{if $Y=\emptyset$ and $\gamma \in R$}\\
                           0      &\text{otherwise.}
                       \end{cases}\\
   [\pc_f(Y)]_1 \, =&     \begin{cases}
                           1      &\text{if $Y=e_j$ and $\gamma = x_j$}\\
                           0      &\text{if $Y=\emptyset$ and $\gamma \in R$}\\
                           0      &\text{otherwise.}
                       \end{cases}\\
   [\pc_f(Y)]_{i>1} \, =&  0.
\end{align*}

Naturally, the induction step has two cases: $f=g+h$ and $f=g\cdot h$.\\ 
{\em Case 1:}  $f=g+h$, then for any $D$
$$[\pc_f(Y)]_D = [\pc_g(Y)]_D +[\pc_h(Y)]_D~~~\forall~D.$$
{\em Case 2:} $f=g\times h$, then for any $D$
$$ [\pc_f(Y)]_D = \sum_{i=0}^d \sum_{j=0}^D [\pc_g(Y^{1,i})]_j[\pc_h(Y^{i+1,d})]_{D-j}$$   
where $Y=\tilde{y}_1,\ldots, \tilde{y}_d$.
The size of the resulting circuit $C'$  is $O(d^3\size(C))$, and $C'$ can in fact be computed in time $O(d^3\size(C))$ given $C$ as the input.  
\end{proof}

\subsection*{Partial Coefficient functions}
For a given commutative  polynomial let $f(X)=\sum_{m}c_m m$, the partial coefficient of a given monomial $m$ (\cite{Mal07}) is a polynomial defined as
$\pcoeff(f,m)=\sum_{m', m|m'} c_{m'} \frac{m'}{m}$.

We extend the above definition to the case of non-commutative polynomials as follows. Let $f$ be non-commutative polynomial, and $m$ a non-commutative monomial. Then
$\pcoeff(f,m){=}\sum_{m'=m\cdot m''} c_{m'}m''$.

The corresponding computational problem can be defined in the following way.
\begin{problem}[Coefficient Polynomial (\pcoeff)]
{\em Input}: A non-commutative arithmetic circuit $C$ computing a polynomial $f$, and a monomial $m$.\\
{\em Output}: A non-commutative arithmetic circuit that computes $\pcoeff(f,m)$. 
\end{problem}
\begin{theorem}
\label{thm:pcoeff}
$\pcoeff$ can be computed in deterministic  time ${\sf poly}(\size(C), n, \deg(m))$.
\end{theorem}
\begin{proof}
The algorithm is similar to the proof of Theorem~\ref{thm:mcoeff}, except that we need to construct
an arithmetic circuit rather than a value. We use the following recursive formulation similar to
(\ref{eq:mc}).

If $f=\alpha\in R\cup \{x_{1},\dots,x_{n}\}$ and $m=\emptyset$ then $\pcoeff(f,m)= \alpha$. For the
summation $f=g+h$ we compute $\pcoeff(f,m) = \pcoeff(g,m) + \pcoeff(h,m)$. The final case to handle
is a multiplication gate. We define shorthand for sets of variables. Let $m=x_{1}\cdots x_{d}$,
$m_{i}=x_{1}\cdots x_{i}$ and $m'_{i} = x_{i+1}\cdots x_{d}$ the rest of the monomial. We define
$m_{0}=\emptyset$. Then
$\pcoeff(f,m)=\sum_{i=0}^{d-1}\mc(g,m_{i})\pcoeff(h,m'_{i})+\pcoeff(g,m)\cdot\pcoeff(h,\emptyset)$.
The rest of the proof is analogous to that of Theorem~\ref{thm:mcoeff} except that, we need to compute and store  the values $\mc(p_v, m_{i,j})$, and $\pcoeff(p_v,m_{i,j} )$ for every gate $v$ in the circuit in a bottom up fashion. 
\end{proof}
 {\em Acknowledgements: }
The authors like to thank V.\ Arvind  and Markus Bl\"aser for helpful discussions and pointing out specific problems to work on.  The authors also thank anonymous referees for their comments which helped in  improving the presentation. 
This work was partially done while  the first author was visiting IIT Madras sponsored by  the Indo-Max-Planck Center for Computer Science. 

\bibliographystyle{abbrv}

\newpage
\appendix
\section{Permanent as a sum over cycle covers}
 Let $G$ be a weighted directed graph on $n$ vertices and   $x_{i,j}$ denote the weight of the directed edge $(i,j)\in E(G)$, the Cayley permanent of $G$ denoted by $\cperm(G)$ is the permanent of the weighted adjacency matrix of $G$.
A cycle cover of $G$ is a collection of vertex disjoint cycles $\pi=(C_1,\ldots C_k)$ that cover all the vertices in $G$. Cayley weight of a cycle cover is the product of weights of edges in the cover, multiplied in the Cayley order, i.e., if $\pi(i)$ denotes the successor of node $i$ in the cycle cover $\pi$, then the weight of the cover $\pi$ is $\prod_{i=1}^{n}x_{i,\pi(i)}$. It is known that~\cite{Bur00} $\cperm(G)$ is the sum of the Cayley weights of all cycle covers of $G$.

\section{Proof Sketch for Proposition~\ref{prop:expsum}}
\label{app:expsum}
Note that the proof given in~\cite{MR13} uses the equivalence of skew circuits with ABPs, which does
not hold in the non-commutative setting. Our argument is similar to the one in~\cite{MR13} except
that we argue over ABPs themselves rather than skew circuits. We give a sketch of the proof here.
Let $P$ be an ABP computing the non-com\-mu\-ta\-tive polynomial given by $p_P(x_1,\ldots, x_n,
y_1,\ldots, y_m)$. Let $0=i_0<i_1<\dots<i_m$ be the layers of $P$ that witness the fact that $P$ is
read-once certified in $Y=\{y_1,\ldots, y_m\}$. Without loss of generality assume that every layer
of $P$ has exactly $r$ nodes.  Let $g^{j}_1,\ldots, g^{j}_r$  be the nodes in the layer $i_j$. Note that variable $y_j$ is read  in layers between $i_{j-1} $ to $i_j$ and is never used beyond that point and no other variable from $Y$ appears in layers between $i_{j-1}$ and $i_{j}$.  Let $P_j$ be the  portion of $P$ consisting only of layers of $P$ from $i_{j-1}$ to $i_{j}$. Let $p_j[\ell , k]$
be the polynomial represented as sum of weights of $g^{j}_{\ell}\leadsto g^{j+1}_k$ paths in $P$. Then
\begin{align*}
  p_P = \sum_{k_1,\ldots, k_m} p_0[s,k_1] p_1[k_1, k_2]\cdots p_{m-2}[k_{m-2},k_{m-1}]\cdot p_{m-1}[k_{m-1},t].
\end{align*}  
For $0\le j< m$, $k,k'\in\{1,\ldots r\}$ and $b\in\{0,1\}$, let $p_{j}^b[k,k'] =
p_{j}[k,k'](y_{j+1}=b)$ where we substitute $y_{j+1}$ with $b$. Then
\[
    \sum_{e_1,\ldots, e_m\in \{0,1\}}p_P(X, e_1,\ldots, e_m)
\]
is equal to
\begin{align*}
  \sum_{k_1,\ldots, k_m} \bigl[\,&(p_{0}^0[s,k_1]+p_{0}^{1}[s,k_1]) \cdot ( p_{1}^0[k_1, k_2]+p_{1}^{1}[k_1, k_2])
  \cdot \cdots \\ 
  &\cdot  (p_{m-2}^{0}[k_{m-2},k_{m-1}]+p_{m-2}^{1}[k_{m-2},k_{m-1}])  \\
  &\cdot  (p_{m-1}^{0}[k_{m-1},t]+ p_{m-1}^{1}[k_{m-1},t])\,\bigr].
\end{align*}
Thus we can take sums of $P_j$ with $y_j=0$ and $y_j=1$ independent of $P_{j'}$ when $j\neq j'$. In
the following we describe this construction and omit the proof of correctness and the bound on the
size of the resulting ABP.  Now create two copies $P_{j}^{0}$ and $P_{j}^{1}$ where $P_{j}^b$ is
obtained by setting $y_j=b$ for $b\in\{0,1\}$. For every $1\le k\le t$, merge the copies of
$g^{j}_k$ in $P_{j}^{0}$ and $P_{j}^{1}$ into a single vertex $h^{j}_k$, and similarly copies of
$g^{j+1}_k$ in $P_{j}^{0}$ and $P_{j}^{1}$ into a single vertex $h^{j+1}_k$. Let
$Q_j$ be the resulting program.  Let $Q$ be the ABP $Q_1\cdot Q_2 \cdots Q_m$ obtained by doing the
following for every $j$: For $1\le k\le t$ glue the copies of $h^{j}_k$ in the top layer of $Q_j$
and first layer of $Q_{j+1}$ to get a single vertex.  Let $Q$ be the resulting ABP.  From the
observation above, we have
\[
    p_Q = \sum_{e_1,\ldots, e_m\in\{0,1\}} p_P(x_1,\ldots, x_n, e_1,\ldots, e_m)
\]
as required. By the construction above, size of $Q$ is at most twice that of $P$. This completes
the proof.

\section{Recap of Gentry's Proof}
For clarity we will repeat the proof of Gentry with some corrections by Goldreich. Readers familiar
with Gentry's Proof can skip this section. We will denote the inverse of an element $g$ in a group
by $g^{-1}$.

\begin{definition}
    Let $R$ be some algebra. A product program $P$ over $R$ with $n$ instructions for an
    input of length $m$ is given by
    $P=\left(
        a_{0},\left(\iota_{1},a_{1,0},a_{1,1}\right),\dots,\left(\iota_{n},a_{n,0},a_{n,1}\right)
    \right)$.
    Where we call the sequence of length $n$ of the form
    $\left( \iota_{i},a_{i,0},a_{i,1}\right)_{i\in [n]} \in ([m],R,R)^{n}$ the instructions
    and $a_{0}$ the starting element.

    It computes on an input $x_{1},\dots,x_{m}\in \{0,1\}$ the product
    \[
        a_{0}\cdot \prod_{i=1}^{n} a_{i,x_{\iota_{i}}}.
    \]
\end{definition}
In words, our product program decides for every instruction $(\iota_{i},a_{i,0},a_{i,1})$ if it
should multiply $a_{i,0}$, if the bit of $x$ at the position $\iota_{i}$ is zero, or $a_{i,1}$, if
the bit of $x$ at the position $\iota_{i}$ is one. We will
generally not distinguish between an input as vector of length $m$ or a string of length $m$ and
will index the string $x$ with $x_{i}$ to mean the $i$th bit.

\begin{lemma}[{\cite[Lemma 2]{Gentry14}}]\label{lem:noncom:gentrylem2}
    For any division algebra $R$, the group of units of $R^{2\times 2}$ contains a subgroup
    isomorphic to $S_{3}$. In particular, $R^{2\times 2}$ contains the matrices
    \[
        r=\begin{pmatrix} 0& -1\\ 1& -1\end{pmatrix}\; \text{and}\;
        s=\begin{pmatrix} 0& 1\\ 1& 0\end{pmatrix}.
    \]
\end{lemma}
\begin{proof}
    It is clear that $R^{2\times 2}$ contains these matrices. It is easy to see that $r$ has
    order three and $s$ order two and hence by Lagrange's Theorem they generate a group of order at
    least six.  We know that the only elements in the group are $1,s,r,r^{2},rs,sr$ as $rs=sr^{2}$
    and $r^{2}s=sr$ and $s$ has order two. As $S_{3}$ is the only non-abelian group of order six the
    lemma holds.
\end{proof}

With this we can now show a product program that outputs one if an assignment satisfies a 3-CNF
formula and zero else.
\begin{lemma}[{\cite[Lemma 3]{Gentry14}}]\label{lem:noncom:gentrylem3}
    There exists a product program of length $2^{d}+2^{d-1}-2$ over the group $S_{3}$ that computes
    a disjunction of $d$ literals. It outputs $1$ if $x$ satisfies the disjunction and $r$
    otherwise.
\end{lemma}
\begin{proof}
    Let $1$ be the multiplicative neutral element of the $R^{2\times 2}$.
    
    We give a proof by induction and assume $a_{0}$ to be $1$. Let $d=1$. As $a_{1,0}=r$ and
    $a_{1,1}=1$ the proof is clear. Let us now assume that the lemma is true for $d-1$ and let
    $P_{d-1}$ be the constructed product program. Let $x_{1},\dots,x_{n}$ be the bits of our input
    and $b_{0}=s$ and $b_{1}=1$. Then we construct the program such that the multiplication will be
    performed as follows:
    \[
        b_{x_{d}}\cdot P_{d-1}(x_{1},\dots,x_{d-1}) \cdot b_{x_{d}} \cdot \left(P_{d-1}(x_{1},\dots,x_{d-1})\right)^{-1}.
    \]
    Here $\left(P_{d-1}(x_{1},\dots,x_{d-1}\right)^{-1}$ is replacing all instructions
    $(\iota_{i},(\alpha,\beta))$ by the corresponding instruction
    $(\iota_{i},(\alpha^{-1},\beta^{-1}))$.
    
    Let us now prove the correctness. If $x_{i}=1$ then it is clear that the layer $i$ evaluates to
    one as one commutes with all elements, especially $P_{d-1}$ and $P_{d-1}^{-1}$. Hence all layers
    above will also evaluate to one as $b_{0}^{2}=b_{1}^{2}=1$.

    If all bits of the input are zero then $P_{d-1}(x_{1},\dots,x_{d-1})=r$ by induction,
    $b_{x_{d}}=s$ and hence $P_{d}(x_{1},\dots,x_{d})=srsr^{-1}$. By the equalities above this is
    equal to
    \begin{equation*}
       s(rs)r^{-1} = s(sr^{2})r^{-1} = s^{2}r = 1r.
    \end{equation*}
\end{proof}

Let
\[
    t=\begin{pmatrix} 1& 0\\ 0& 0\end{pmatrix}.
\]
\begin{theorem}[{\cite[Theorem 6]{Gentry14}}]\label{thm:noncom:gentrythm6}
    For any division algebra $R$ and any constant $d$ one can construct a product program of
    length $k(2^{d}+2^{d-1}-2)$ for a $d$-CNF formula with $k$ clauses. It outputs $t$ if the
    formula is satisfied by $x$ and $0$ otherwise.
\end{theorem}
\begin{proof}
    Let our clauses be given by $c_{1},\dots,c_{k}$. By Lemma~\ref{lem:noncom:gentrylem3} we get a
    product program $P_{c_{i}}$ for every clause $c_{i}$. Then we construct our product program $P$
    to compute the multiplication as follows:
    \[
        \left(\prod_{i=1}^{k} t\cdot P_{c_{i}}(x_{1},\dots,x_{m})\right)t.
    \]
    Here the multiplication with $t$ can easily be simulated with an instruction of the form $(1,(t,t))$.
    
    For ease of notation we have written the complete variable set for our clause product programs
    but we can easily remove unneeded variables from a clause.

    Let us give a correctness argument. Suppose one of the clauses is not satisfied. This then
    contributes a value of $t\cdot r$ to our product. It can be seen that this program has only two
    possible outcomes if one equation is not fulfilled. Namely,
    \begin{align*}
      t \cdot r &= \begin{pmatrix} 0& -1\\ 0& 0\end{pmatrix}\\
      \intertext{and for any $g$}\\
      t \cdot r \cdot t \cdot g&= \begin{pmatrix} 0& 0\\ 0& 0\end{pmatrix},
    \end{align*}
    as $t\cdot r\cdot t=0$. Notice, that we will always multiply the value $t\cdot r$ from an
    unsatisfied clause with $t$ to the right. Hence, by associativity the resulting matrix will
    always be zero.

    If all clauses are fulfilled the value computed is $t^{k+1}$. However, $t$ is idempotent in
    $R^{2\times 2}$ and hence is equal to $t$.
\end{proof}

Let $P$ be the product program as in Theorem~\ref{thm:noncom:gentrythm6}. Then it is obvious that
\[
    \sharpsat(\phi) = \left(\sum_{e\in \{0,1\}^{m}} P_{\phi}(e)\right)_{(1,1)},
\]
the first entry in the resulting matrix, as every satisfying assignment contributes exactly one and
every unsatisfied assignment zero. To compute the sum with the Cayley determinant we will use the
following special matrix form.
\begin{definition}
    We say a $n\times n$ matrix $M$ is a \emph{barber pole matrix} if it is of the form
    \[
        M[i,j]=\begin{cases}
        \alpha_{i} &\text{if $i=j$,}\\
        \beta_{i}  &\text{if $i=j+1\mod n$,}\\
        0         &\text{otherwise.}
        \end{cases}
    \]
    for $\alpha_{i},\beta_{i}$ non zero.
\end{definition}

Notice, that for every barber pole matrix there exists only two cycle covers. Either the one where
every vertex takes a self-loop or the single cycle.
\begin{theorem}[{\cite[Theorem 5]{Gentry14}}]\label{thm:noncom:gentrythm5}
    The value $\sum_{e\in \{0,1\}^{m}} P_{\phi}(e)$ can be computed by the Cayley determinant of a
    matrix of size $n\times n$ over the algebra $R^{2\times 2}$.
\end{theorem}
\begin{proof}
    Let $P=\left(1,\left(\iota_{1},a_{1,0},a_{1,1}\right),\dots,\left(\iota_{n},a_{n,0},a_{n,1}\right)\right)$. Let
    \[
        \mathcal{I}_{\ell} = \{i\in [n] \mid \text{the $i$th instruction uses the $\ell$th bit of the
            input} \}.
    \]
    Let $\mathcal{I}_{\ell}$ have the instructions
    $i_{\ell,1},\dots,i_{\ell,\lvert \mathcal{I}_{\ell}\rvert}$. Let $\pi_{0}$ be the identity
    permutation and
    $\pi_{1}(i_{\ell,\kappa})=i_{\ell,\kappa +1 \mod \lvert \mathcal{I}_{\ell}\rvert }$ the
    ``shifted'' permutation. Notice, that this corresponds to multiple cyclic permutation,
    consisting of cycles of length $\lvert \mathcal{I}_{\ell}\rvert$ for $2\leq \ell \leq n$ where
    the elements of the cycle are the elements in $\mathcal{I}_{\ell}$. They are ordered in the
    natural order of the instructions. Furthermore, we pad every set $\mathcal{I}_{\ell}$ to have
    size at least two.

    Then we define the matrix
    \[
        M[i,j] = \begin{cases}
            (-1)^{\lvert \mathcal{I}_{k}\rvert-1} a_{i,0}  &\text{ if $i=i_{k,1}$, $j=i_{k,2}$ for
                some $k$,}\\
            a_{i,b}                                     &\text{ otherwise, if $j=\pi_{b}(i)$,}\\
            0                                          &\text{ otherwise.}
        \end{cases}
    \]
    Left to show is that $\cdet(M)$ is indeed computing the value of the product program. Let us
    look at this matrix a bit closer and assume the position $M[i_{k,1},i_{k,2}]$ is $a_{i,0}$. We can
    see that this is a block barber pole matrix where the entry not on the diagonal are permuted.
    
    Let us generate the matrix only for the set of instructions
    $\mathcal{I}_{\ell} = \{i_{\ell,1},\dots,i_{\ell,\lvert \mathcal{I}_{\ell}\rvert}\}$. At first
    we add all the entries for $\pi_{0}$ which are just on the diagonal. If we now look at
    $\pi_{1}$, we see that the first entry we add is at position $(1,i_{\ell,2})$ where $i_{\ell,2}$
    is the index of the next instruction. We continue this until $\pi_{1}$ wraps around. This is
    clearly a cycle in the graph represented by the matrix.

    It is now clear that $\pi_{1}$ produces a cycle for every $\mathcal{I}_{\ell}$. In essence it
    enforces that we either take all self-loops or all elements corresponding to $\pi_{1}$. Meaning
    we either multiply all values in the instructions asking $x_{\ell}$ for $x_{\ell}$ being zero or
    all values for the instructions where $x_{\ell}$ is one.

    By this argument it is clear that one cycle cover is the value of the product program where we
    chosen every bit of the input and hence the value of all cycle cover is
    $\sum_{e\in\{0,1\}^{m}} P(e)$. Now it is easy to see that the actual value of
    $M[i_{k,1},i_{k,2}]$ exactly cancels the sign the determinant introduces.
\end{proof}

As the Cayley permanent is equal to the Cayley determinant for this construction we get the
following corollary.
\begin{corollary}\label{cor:noncom:gentryperm}
    The value $\sum_{e\in \{0,1\}^{m}} P_{\phi}(e)$ can be computed by the Cayley permanent of a
    matrix of size $n\times n$ over the algebra $R^{2\times 2}$ by removing the scalar factor of $(-1)^{\lvert\mathcal{I}_{k}\rvert-1}$.
\end{corollary}

\end{document}